\documentclass[11pt]{article}

\usepackage{geometry}
\usepackage[english]{babel}
\usepackage{lmodern}
\usepackage{microtype}

\usepackage{graphicx}
\usepackage{authblk}
\usepackage{hyperref}
\usepackage{xcolor}
\usepackage{standalone}
\usepackage{tikz}
\usepackage{amsmath,amsfonts,amssymb}
\usepackage{amsthm}
\usepackage[capitalise]{cleveref}
\usepackage[shortlabels]{enumitem}
\usepackage{caption}
\usepackage{subcaption}

\graphicspath{{img/}}

\usepackage{tabularx}

\newcommand{\N}{\ensuremath{\mathbb{N}}}
\newcommand{\Z}{\ensuremath{\mathbb{Z}}}

\newcommand{\R}{\ensuremath{\mathbb{R}}}

\renewcommand{\epsilon}{\varepsilon}

\renewcommand{\vec}[1]{\mathbf{#1}}

\newtheorem{theorem}{Theorem}
\newtheorem*{theorem*}{Theorem}
\newtheorem{lemma}[theorem]{Lemma}
\newtheorem*{lemma*}{Lemma}

\newtheorem{coro}[theorem]{Corollary}
\crefname{coro}{Corollary}{Corollaries}

\crefname{prop}{Property}{Properties}
\newtheorem*{prop*}{Property}

\crefname{propo}{Proposition}{Propositions}
\newtheorem*{propo*}{Property}

\theoremstyle{definition}
\newtheorem{definition}{Definition}

\theoremstyle{remark}
\newtheorem{remark}{Remark}
\crefname{remark}{Remark}{Remarks}

\crefname{ex}{Example}{Examples}

\title{Pattern Complexity of Aperiodic Substitutive Subshifts}
\date{}

\author[1]{Etienne Moutot}
\author[2]{Coline Petit-Jean}

\affil[1]{\href{mailto:etienne.moutot@lis-lab.fr}{etienne.moutot@lis-lab.fr}\\
Aix-Marseille Université, CNRS, LIS, France}
\affil[2]{\href{mailto:coline.petit-jean@ens-lyon.fr}{coline.petit-jean@ens-lyon.fr}\\
ENS de Lyon, France}

\begin{document}

\maketitle

\begin{abstract}
This paper aims to better understand the link better understand the links between aperiodicity in subshifts and pattern complexity.
Our main contribution deals with substitutive subshifts, an equivalent to substitutive tilings in the context of symbolic dynamics.
For a class of substitutive subshifts, we prove a quadratic lower bound on their pattern complexity. 
Together with an already known upper bound, this shows that this class of substitutive subshifts has a pattern complexity in $\Theta(n^2)$.
We also prove that the recent bound of Kari and Moutot, showing that any aperiodic subshift has pattern complexity at least $mn+1$, is optimal for fixed $m$ and $n$.
\end{abstract}

\section{Introduction}

One of the most fascinating aspects of tilings of the plane is the existence of \emph{aperiodic tilesets}.
That is, a set of tiles that tessellates the plane but only in a non-periodic manner.
Aperiodic tilesets even exit in the specific setting of Wang tilings, where tiles are square with colored edges and cannot be rotated nor reflected. Two tiles can then be placed next to each other if the colors of the matching edges are the same.
Interestingly, aperiodic tilesets do not exist for a similar model in dimension one. If one considers tilings of the infinite line $\Z$ by bi-color Wang dominoes, then any set of dominoes tiling the line also tiles it periodically.
This has a consequence on the decidability of the most fundamental problem about tilings: the domino problem.
This problem just asks if a given tileset can or cannot tile the plane.
Considering this problem on a line or a plan changes everything: the domino problem for bi-infinite line is decidable in polynomial time, as it is equivalent to the existence of a cycle in a finite graph; whereas it is undecidable if one considers tilings of the plane.
From now on, ``tiling'' will always mean ``Wang tiling'', except stated otherwise. 

The first aperiodic tileset is due to Berger in 1966 \cite{Berger}, who disproved Wang's conjecture stating that aperiodic tilesets of the plane could not exist.
His tileset was expressed as a set of Wang tiles, a model introduced by Wang to study fragments of first order logic \cite{Wang}. Each tile is a unit square with a color on each edge. 
Two tiles can then be placed next to each other if the colors of the matching edges are the same.
Then from a finite set of tiles, one tries to tile the plane by copies of these tiles without rotations.
Berger set of tiles was initially made of 20426 tiles, reduced to 104 in his PhD thesis \cite{BergerPhD}.
This started the search for the smallest aperiodic set of Wang tiles.
Among many, one can cite Robinson's tileset \cite{Robinson} which has only 56 tiles.
With a radically different idea, Kari built a 14 tiles aperiodic tileset \cite{Kari_1996}.
His construction was quickly improved by Culik to 13 tiles \cite{Culik_1996}, which held the record for many years.
Finally, in 2015 Jeandel and Rao found an 11 tiles aperiodic tileset, and proved that it was the smallest possible \cite{Jeandel_Rao}.

If it is interesting to have aperiodic tilesets with few tiles, the number of tiles is not a good estimation of how ``complex'' the resulting tilings can be.
Indeed, with the same number of tiles one can end up with very simple tilings as well as extremely complex (encoding for example the space-time diagram of a Turing machine).
A better way of estimating how complex a tiling is by its \emph{pattern complexity}.
For a tileset $T$, we denote by $\mathcal{L}_{m,n}(T)$ the set of all $m\times n$ rectangular patterns appearing in at least one valid tiling by $T$. The pattern complexity of $T$ is then the cardinal of this set.
Besides to carrying more information about how ``complicated'' a tiling is, this definition of complexity can be seen as a generalization of the number of tiles, as the latter is simply the number of patterns of size $1\times 1$.
Computing the exact complexity of a tileset if often complicated, and a more studied quantity is the topological entropy of a tileset. 
%
However, topological entropy is a very rough estimation of the complexity, and provide only an asymptotic estimation of the growth of the complexity function.
In this paper, we are interested in bounding the exact complexity of aperiodic tilesets.
More precisely, we want to understand what is the minimal complexity that can be achieved by an aperiodic tileset; or equivalently, how ``simple'' can an aperiodic tileset be.
In \cite{Kari_Moutot_2020}, it is shown that any aperiodic tileset have complexity at least $mn + 1$.
The first result of the current paper shows that this bound is optimal for fixed $m$ and $n$.

\smallskip

One of the most common techniques to build an aperiodic tileset is to use two dimensional substitutions.
In most of the cases, the resulting tilings are close to being fixpoints of well-chosen two-dimensional substitutions, leading to their aperiodicity.
This is the case of Berger's tileset, Robinson's, and even Jeandel and Rao's tileset was recently found to have a substitutive structure \cite{Labbe_2021}. Their substitutive nature also ensures that they have zero topological entropy.
Studying the complexity of substitutive tilesets is therefore a way of understanding the complexity of many aperiodic tilesets.
One of the rare exceptions is Kari's aperiodic tileset, which is known not to be substitutive, as it was shown to have positive topological entropy \cite{Durand_Gamard_Grandjean_2014}.
This also shows that it has an exponential pattern complexity.

More generally, \emph{substitutive subshifts} have been extensively studied in dimension one \cite{DurandSubs, Fogg}.
A substitutive subshift is a set of colorings of the infinite line $\Z$ generated by infinite iterations of a one-dimensional substitution.
In dimension one, Pansiot fully classified substitutions in terms of factor complexity \cite{Pansiot}.

In the two-dimensional case, substitutions have been mostly studied in the context of geometrical tilings \cite{ARNOUX2007251, Fernique, Grunbaum_Shephard_1987, Penrose, solomyak_nonperiodicity_1998}.
In this paper, we focus on discrete substitutions over two-dimensional words. 
Our second main result is a quadratic lower bound on the pattern complexity of a large class of two dimensional substitutive subshifts.
Together with Robinson's upper bound on the complexity of two-dimensional substitutions \cite[Theorem 7.16]{Arthur_Robinson_2004}, this shows that the pattern complexity of subshifts from this class is in $\Theta(n^2)$.
Such a bound paves th{}e way towards a classification of two-dimensional substitutions in terms of pattern complexity, as Pansiot did for one-dimensional substitutions.

After introducing definitions and basic properties in \cref{sec:def}, we show in \cref{sec:optimal} that the $mn+1$ complexity bound is optimal for fixed $m,n$.
Then we prove our lower bound on our class of substitutive subshifts in \cref{sec:lowerbound}.
Finally we conclude by some remarks on possible directions for future work in \cref{sec:future}.


\section{Preliminaries}
\label{sec:def}
In this section, we introduce all the useful definitions and properties.

We will use the following conventions for our notations.
Capital $A$ designates a finite alphabet and $a$ and $b$ elements of it.
Capital $X$ is a subshift, whose configurations are usually denoted by the letter $c$, and patterns by $p$.
Vectors of $\Z^2$ are bold $\vec u, \vec v \in \Z^2$.
$\sigma$ is a substitution, and $i,j,k,l,m,n$ designate integers.

\subsection{Subshifts}
Tilings have an equivalent definition as subshifts of finite type.
Let $A$ be a finite alphabet. A \emph{pattern} of support $D\subset \Z^2$ is a coloring of $D$ with colors of $A$, i.e. $p\in A^D$, and a \emph{configuration} $c\in A^{\Z^2}$ is a coloring of $\Z^2$.
We say that a pattern $p$ of support $D$ \emph{appears} in a pattern or a configuration $p'$ if there exists $\vec u \in\Z^2$ such that for all $\vec v\in D$, $p_{\vec v} = p'_{\vec v - \vec u}$, in this case we write $p\sqsubseteq p'$ (and $p\sqsubset p'$ when $p\neq p'$). 
Let $F$ be a set of patterns, a \emph{subshift} $X_F$ is a set of configurations in which none of the patterns of $F$ appear:
\[ X_F := \{ c\in A^{\Z^2} \mid \forall p\in F, p\not\sqsubset c  \} . \]
If $F$ is finite, $X_F$ is called a \emph{subshift of finite type} (SFT for short).
Subshifts have an equivalent definition as subsets of $A^{\Z^2}$ that are shift-invariant and topologically closed.
In other words, any shift of a configuration of a subshift $X$ is still in $X$, and the limit of a sequence of shifts of configurations of $X$ is in $X$.

Let $\tau_{\vec u}$ denote the shift action by the vector $\vec u\in\Z^2$: let $p$ be a pattern or configuration of support $D$ (which is $\Z^2$ if $p$ is a configuration), then for all $\vec v\in D+\vec u, \tau_{\vec u}(p)_{\vec v} = p_{\vec v - \vec u}$.
A~configuration $c$ is said to be \emph{periodic} if there exists $\vec u\in\Z^2\backslash\{(0,0)\}$ such that $c = \tau_{\vec u}(c)$.
A~subshift is \emph{aperiodic} if it is not empty and contains no periodic configuration.

The set of valid tilings by a set of Wang tiles mentioned in the introduction is an SFT. These models are actually equivalent, as any SFT can be made into a set of Wang tilings by application of a single morphism.
As subshifts are the most natural model to deal with substitutions, we chose to use this formalism in this paper instead of Wang tilesets.

\subsection{Pattern complexity}

The set of all patterns appearing in a configuration $c$ is called the language of $c$, $\mathcal{L}(c)$; and we denote by $\mathcal{L}_{m,n}(c) := \mathcal{L}(c) \cap A^{\{0,\hdots, m-1\}\times\{0,\hdots, n-1\}}$ the set of all $m\times n$ rectangular patterns appearing in $c$. 
Similarly, $\mathcal{L}(X) = \bigcup_{c\in X} \mathcal{L}(c)$ is the set of all patterns of the subshift $X$ and $\mathcal{L}_{m,n}(X) := \bigcup_{c\in X} \mathcal{L}_{m,n}(c)$ is the set of $m\times n$ patterns of $X$.
The \emph{pattern complexity} (or just \emph{complexity}) of a configuration (resp. a subshift) is the number of such rectangular patterns $|\mathcal{L}_{m,n}(c)|$ (resp. $|\mathcal{L}_{m,n}(X)|$).

For configurations of dimension one, periodicity and complexity are linked: the Morse-Hedlund theorem states that a one-dimensional configuration $w\in A^\Z$ is periodic if and only if there exists n such that $P_c(n) \leq n$ \cite{Morse_Hedlund}.
For one-dimensional SFTs, the question is trivial as every non-empty SFT contains a periodic configuration (see for example \cite{Lind_Marcus}), and thus the complexity of the subshift cannot have any impact on its aperiodicity.

In dimension two, the link between complexity and periodicity becomes much more involved. 
Nivat conjectured in 1997 an implication similar to Morse and Hedlund's: any two-dimensional configuration with $P_c(m,n)\leq m,n$ for some $m,n$ must be periodic \cite{Nivat}. 
The other direction does not hold, as there exists periodic configurations with exponential complexity \cite{Berthe2000}.
Since two dimensional aperiodic SFTs do exist, it is natural to try to link their complexity with their aperiodicity, and a bound relating the two has been recently found:
\begin{theorem*}[Kari, Moutot \cite{Kari_Moutot_2020}]
	\label{th:karimoutot}
	Let $X$ be a subshift such that $\exists c\in X, \exists m,n, {|\mathcal{L}_{m,n}(c)| \leq mn}$.
	Then there exists $d\in X$ that is periodic.
\end{theorem*}
Which can be rephrased as a lower bound on the complexity of aperiodic subshifts.
\begin{coro}
	\label{coro:karimoutot}
	Let $X$ be an aperiodic subshift.
	Then for all $c\in X$ and all $m,n$, \[ {|\mathcal{L}_{m,n}(c)| \geq mn+1} . \]
\end{coro}

\subsection{Substitutions}
Substitutions on words have been extensively studied, see for example \cite{Fogg} for a good reference about them.
In dimension two, they have been studied mostly in the context of substitution of $\R^2$, that is, substitutions producing geometrical tilings.
In the current paper, we restrict ourselves to discrete uniform two-dimensional substitutions.
In this context, a \emph{substitution} of size $m\times n$ (or support $\{0,\hdots,m-1\}\times \{0,\hdots,n-1\}$) is a map $\sigma:A\rightarrow A^{\{0,\hdots,m-1\}\times \{0,\hdots,n-1\}}$.
It is said to be \emph{primitive} if there exists $k$ such that for all $a,b\in A,$ we have $a\sqsubseteq \sigma^k(b)$, and \emph{invertible} if no two letters have the same image.
The \emph{subshift of a substitution}, also called a substitutive subshift, is the set of configurations that can be made from applying the substitution to letters of $A$:
\[ X^\sigma := \left\{ c\in A^{\Z^2} \mid \forall p\sqsubset c, \exists a\in A, \exists k\in\N, p\sqsubseteq \sigma^k(a) \right\} . \]
In general this subshift is only a sofic subshift and not an SFT, as shown by Mozes \cite{Mozes_1988}.

Substitutions are a very convenient way of generating aperiodicity, by example by considering fixpoints configurations.
In the context of geometrical tilings, Solomyak characterized when the subshift of a substitution is aperiodic or not \cite{solomyak_nonperiodicity_1998}.
We recall here his result in the specific case of uniform rectangular substitutions.

\begin{definition}[Unique desubstitution]
	A $\mathcal{D}$-pattern $p\in A^{\mathcal{D}}$ can be \emph{uniquely desubstitued} by a substitution $\sigma$ of size $(k,l)$ if there exists $\vec t\in \Z^2$ which is unique modulo $(k,l)$ such that there exists $c\in X^\sigma, \tau_{\vec t}(\sigma(c))|_{\mathcal{D}} = p$.
\end{definition}

\begin{theorem}[Solomyak \cite{solomyak_nonperiodicity_1998}]
	\label{theo:solomyak}
	Let $\sigma$ be a primitive substitution.
	Then $X^\sigma$ is aperiodic if and only if there exists $\rho>0$ such that all patterns $p$ whose support fits in $\{ 1,\hdots,\rho \}^2$ can be uniquely desubstituted by $\sigma$.
\end{theorem}

A quadratic upper bound is already known for substitutive subshifts, due to Hansen and Robinson.
\begin{theorem*}[Hansen, Robinson {\cite[Theorem 7.17]{Arthur_Robinson_2004}}]
    Let $\sigma$ be a primitive and invertible substitution with square support.
    Then there exists $K>0$ such that 
    \[ |\mathcal{L}_{n,n}(X)|\leq K n^2 .\]
\end{theorem*}

Our lower bound will not deal with all invertible substitutions, but only with the ones having what we call a \emph{determining position}, which is a position in the image that allows us to uniquely recover the antecedent.
\begin{definition}[Determining position]
    \label{def:det-pos}
	A substitution $\sigma$ has a \emph{determining position} $(i,j)$ if for all $a,b\in A$:
	\[ \sigma(a)_{i,j} = \sigma(b)_{i,j} \Rightarrow a = b .\]
\end{definition}
Equivalently, this means that $\left\{ \sigma(a)_{i,j} \mid a\in A \right\} = A$. See \cref{fig:ex_det-pos} for an illustration of this property.

\begin{figure}[ht]
    \centering
    \includegraphics[scale=0.3]{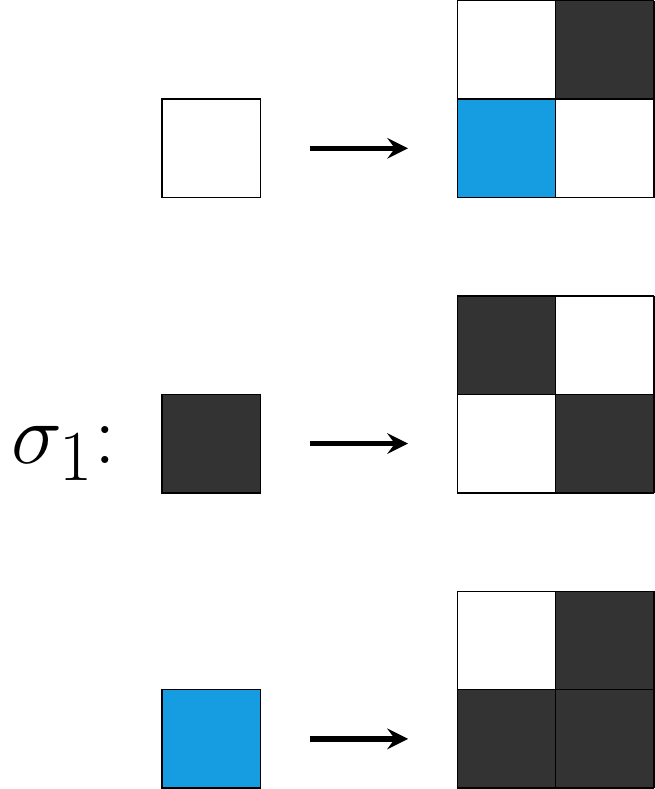}
    \qquad\qquad
    \includegraphics[scale=0.3]{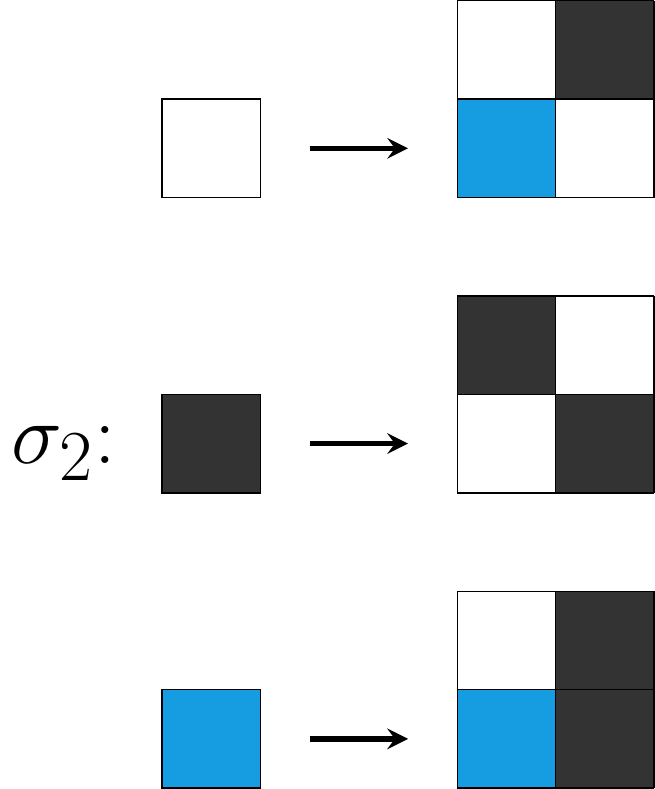}
    \caption{Illustration of \cref{def:det-pos}: $\sigma_1$ has a determining position $(0,0)$ whereas $\sigma_2$ has no determining position.}
    \label{fig:ex_det-pos}
\end{figure}

\begin{remark}
    If $|A| = 2$, all non-constant substitutions have a determining position.
    Indeed, as the images of the two letters must be different, they must differ in at most one position, which is then a determining position.
\end{remark}


\section{Optimality of the \texorpdfstring{$mn+1$}{mn+1} bound}
\label{sec:optimal}

\cref{coro:karimoutot} shows that all configurations in aperiodic subshifts must have complexity at least $mn+1$ for all $m,n$.

It is not known whether this bound is optimal in general, however if $m$ and $n$ are fixed, it is possible to prove its optimality.

\begin{theorem}
	\label{th:optimal}
	For all $m,n \in \N$, there exists an aperiodic SFT $X$ such that for all $c \in X,$ we have $|\mathcal{L}_{m,n}(c)| = mn + 1$.
\end{theorem}

\begin{proof}
Let $m,n\in\N$ be fixed and $Y\subseteq A^{\Z^2}$ be an aperiodic SFT on some alphabet $A$ (for example the set of all Robinson tilings \cite{Robinson}). We will ``blow up'' $Y$ by encoding each of its colors into a large rectangle over alphabet $\{0,1\}$.

Let $k$ be such that $2^{k^2}\geq |A|$.
Then, for each $a\in A$, it is easy to build a $k\times k$ square $S(a)\in\{0,1\}^{\{0,\hdots,k-1\}^2}$ encoding $a$ in such a way that $S$ is an isomorphism.
$S$ can be naturally extended to patterns or configurations of $A^{\Z^2}$.

Let $R_0$ and $R_1$ be the $m\times n$ rectangles with respectively a $0$ and a $1$ in position $(0,0)$ and 0s elsewhere. We define the substitution
\[ \sigma:\begin{cases}
 0 \mapsto R_0\\
 1 \mapsto R_1
\end{cases} .\]

\begin{figure}[ht]
    \centering
    \includegraphics[scale=0.4]{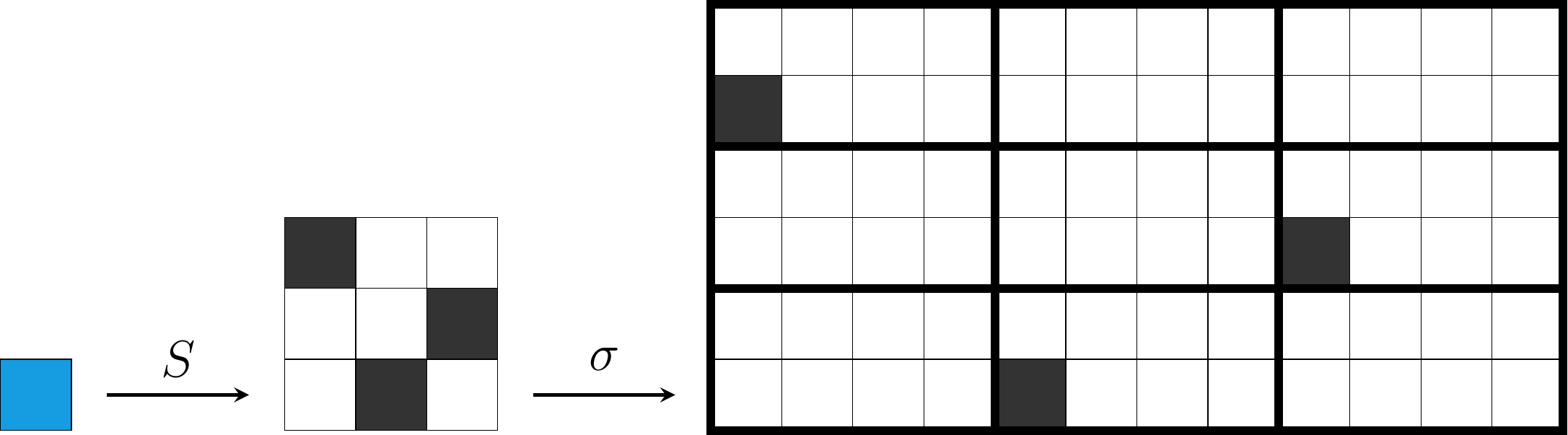}
    \caption{Illustration of the sparse encoding of \cref{th:optimal} with $k=3, m=4$ and $n=2$. We first apply $S$ to some $s\in A$, and then apply $\sigma$ to this image.}
    \label{fig:mn+1_encoding}
\end{figure}

Let $X := \sigma(S(Y)) = \{ \sigma(S(c)) \mid c\in Y \}$.
To prove that $X$ is an SFT, let $F_0$ be the finite set of fobidden patterns of $Y$.
Let 
\[ F_1 = \left\{ p\in \{0,1\}^{\{0,\hdots,2km-1\}\times\{0,\hdots,2kn-1\}} \mid \forall q\in A^{\{0,\hdots,3\}\times\{0,\hdots,3\}},  p \not\sqsubset \sigma(S(q)) \right\}  \]
bet the set of $2km \times 2kn$ patterns not appearing in any $2km \times 2kn$ patterns of the form $\sigma(S(q))$.
Let $F = F_1 \cup \sigma(S(F_0))$.
Then, $X = X_{F}$. Indeed, $F_1$ ensures that all configurations of $X_F$ are desubstituable by $\sigma(S)$, and together with $\sigma(S(F_0))$, we have $c\in X_F \Leftrightarrow c\in \sigma(S(Y))$ thanks to the fact that $\sigma(S)$ is one-to-one.
As $F$ is finite, $X$ is an SFT.

Now, since every $c\in X$ can be written $\sigma(d)$ with $d\in S(X_0)$, every $m\times n$ rectangular pattern of $c$ contains at most one $1$, therefore $|\mathcal{L}_{m,n}(c)| \leq mn+1$.
But because $X$ is aperiodic, \cref{coro:karimoutot} ensures that $|\mathcal{L}_{m,n}(c)| \geq mn+1$, so $|\mathcal{L}_{m,n}(c)| = mn+1$.
\end{proof}

One might hope to improve the $mn+1$ bound of \cref{coro:karimoutot} by using a result of Cassaigne:

\begin{theorem*}[Cassaigne, direct consequence of \cite{Cassaigne_1999}]
\begin{sloppypar}
	Let $c\in A^{\Z^2}$ be such that for all $m,n, {|\mathcal{L}_{m,n}(c)| = mn +1}$.
	then there exists $d\in\overline{\mathcal{O}(c)}$ which is uniform.
	In particular, the subshift $\overline{\mathcal{O}(c)}$ is not periodic.
\end{sloppypar}
\end{theorem*}
\noindent
However, it gives an information about configurations with $|\mathcal{L}_{m,n}(c)|=mn+1$ for \emph{all} $m,n$.
In order to improve \cref{coro:karimoutot} to $mn+2$, one would need a similar result with the hypothesis of the existence of $m,n$ such that $|\mathcal{L}_{m,n}(c)| = mn+1$.
Together with \cref{th:optimal}, this emphasizes the importance of the quantifiers and their order when studying the complexity of subshifts.

The SFT built in \cref{th:optimal} depends on $m$ and $n$, so it does not rule out the possibility of improving the uniform bound in $m$ and $n$.
In the next section we show that the uniform lower bound can indeed be improved for a large class of substitutive subshifts.


\section{Lower bound for substitutive subshifts}
\label{sec:lowerbound}

In this section, we show our lower bound for a class of two-dimensional substitutive subshifts.
First, we need a classical lemma on substitutive subshifts, whose (short) proof is included for sake of completeness.

\begin{lemma}
	\label{lemma:primitive}
	Let $\sigma$ be a primitive substitution and $c\in X^\sigma$.
	Then for all $p\sqsubset c, \sigma(p) \sqsubset c$.
\end{lemma}

\begin{proof}
	Let $p \sqsubset c$. By definition of $X^\sigma$, there exists $a \in A$, $k \in \N$ such that $p \sqsubset \sigma^k(a)$. This implies that $\sigma(p) \sqsubset \sigma^{k+1}(a)$. 
	Moreover, $\sigma$ is primitive: there exists some $l \in \N$ such that for all $b \in A$, $ a \sqsubset \sigma^l(b)$ and consequently $\sigma^{k+1}(a) \sqsubset \sigma^{l + k +1}(b)$.

	By definition of $X^\sigma$, for a big enough pattern $p'$ appearing in c, there is some $b$ such that $\sigma^{l+k+1}(b) \sqsubset p'$. 
	Eventually, we have:
	\[ \sigma(p) \sqsubset \sigma^{k+1}(a) \sqsubset \sigma^{l + k + 1}(b) \sqsubset p' \sqsubset c . \]
\end{proof}

Then, we take a look at the determining position property for composition of substitutions.

\begin{lemma}
    If $\sigma$ and $\sigma'$ have determining positions, then $\sigma \circ \sigma'$ has a determining position.
\end{lemma}

\begin{proof}
Let $\sigma$ and $\sigma'$ be two substitutions of size $(m,n)$ having determining positions $(k,l)$ and $(k',l')$ respectively. Let $a, a' \in A$. Then,
\begin{align*}
& (\sigma \circ \sigma')(a)_{(mk + k', nl + l')} = (\sigma \circ \sigma')(a')_{(mk + k',nl + l')}   \\ 
\Rightarrow& \sigma\left(\sigma'(a)_{(k,l)}\right)_{(k',l')} = \sigma\left(\sigma'(a')_{(k,l)}\right)_{(k',l')} \\
\Rightarrow& \sigma'(a)_{(k,l)} = \sigma'(a')_{(k,l)}\\
\Rightarrow& a = a' .
\end{align*}
\end{proof}

An immediate consequence is the following.

\begin{coro}
	\label{lemma:determining}
    If $\sigma$ has a determining position, then for all $k\geq 1, \sigma^k$ has a determining position.
\end{coro}

The next technical lemma is the key of our lower bound.
Essentially, it provides us a lower bound on the complexity based on the number of times $k$ that a patterns can be uniquely desubstituted.

\begin{lemma}
\label{turbolemme}
	Let $\sigma$ be a substitution of size $(M,M)$ having a determining position and $k \in \N$. 
	For all $n \in \N$, if all the patterns of size $(n, n)$ of $X^\sigma$ have a unique desubstitution by $\sigma^k$, then for all $c\in X^\sigma$:

	\[ |\mathcal{L}_{n,n}(c)| \geq n^2 \left( 1 + \frac{1}{\lceil \frac{n}{M^k} \rceil^2} \right) . \]
\end{lemma}

\begin{proof}
Let us denote by $\mathcal{P}_c(k,l)$ the elements of $\mathcal{L}_{k,l}(c)$ for this proof for better reading of the equations.
Using the fact that $\sigma^k$ also has a determining position we build an injective function $f$ from the disjoint union $\bigsqcup_{0 \leq i,j \leq M^k-1} \mathcal{P}_c(\lceil \frac{n-i}{M^k} \rceil, \lceil \frac{n-j}{M^k} \rceil)$ to $\mathcal{P}_c(n,n)$. The existence of such an injective function, put together with the result of Kari and Moutot (\cref{coro:karimoutot}), yields the result.

\bigskip

Let $(d_1, d_2)\in\{0,\hdots, M^k-1\}$ be the determining position of $\sigma^k$ obtained in \cref{lemma:determining}, ${i,j\in\{ 0, \hdots, M^k-1  \}}$ and ${p \in \mathcal{P}_c(\lceil \frac{n-i}{M^k} \rceil, \lceil \frac{n-j}{M^k} \rceil)}$. The fact that $p$ appears in $c$ implies that it can be extended to bigger patterns also appearing in $c$.
Then, for each such $p$, we fix one $p^{ext} \in A^{\{-1, \dots, \lceil \frac{n-i}{M^k} \rceil + 1 \} \times \{-1, \dots, \lceil \frac{n-j}{M^k} \rceil + 1 \}}$ such that $p^{ext}|_{ \{0, \dots, \lceil \frac{n-i}{M^k} \rceil - 1\} \times \{0, \dots, \lceil \frac{n-j}{M^k} \rceil - 1\}} = p$ and $p^{ext} \sqsubset c$. We define the following pattern:
\[ f_{i,j}(p) = \tau_{i-d_1, j-d_2}(\sigma^k(p^{ext}))|_{ \{0,\dots, n-1\}^2} \]
This is well defined because the fact that $i, j, d_1, d_2 \in \{0, \dots, M^k-1\}$ implies that \[ \{0,\dots, n-1\}^2 \subset (i-d_1, j-d_2) +  \{-1, \dots, \left\lceil \frac{n-i}{M^k} \right\rceil + 1 \} \times \{-1, \dots, \left\lceil \frac{n-j}{M^k} \right\rceil + 1 \} . \]
By \cref{lemma:primitive} it holds that $\sigma^k(p^{ext}) \sqsubset c$ and consequently that $f_{i,j}(p) \in \mathcal{P}_c(n,n)$. 

\begin{figure}[ht]
     \centering
     \begin{subfigure}[b]{0.14\textwidth}
         \centering
         \includegraphics[scale=0.4]{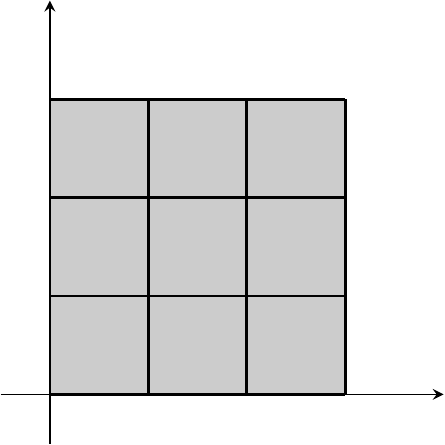}
         \caption{$p$}
         \label{fig:p}
     \end{subfigure}
     \hfill
     \begin{subfigure}[b]{0.24\textwidth}
         \centering
         \includegraphics[scale=0.4]{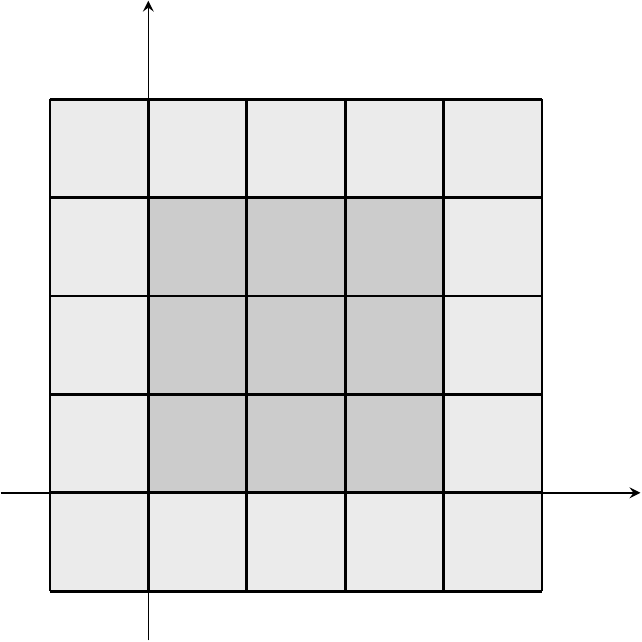}
         \caption{$p^{ext}$}
         \label{fig:pext}
     \end{subfigure}
     \hfill
     \begin{subfigure}[b]{0.55\textwidth}
         \centering
         \includegraphics[scale=0.8]{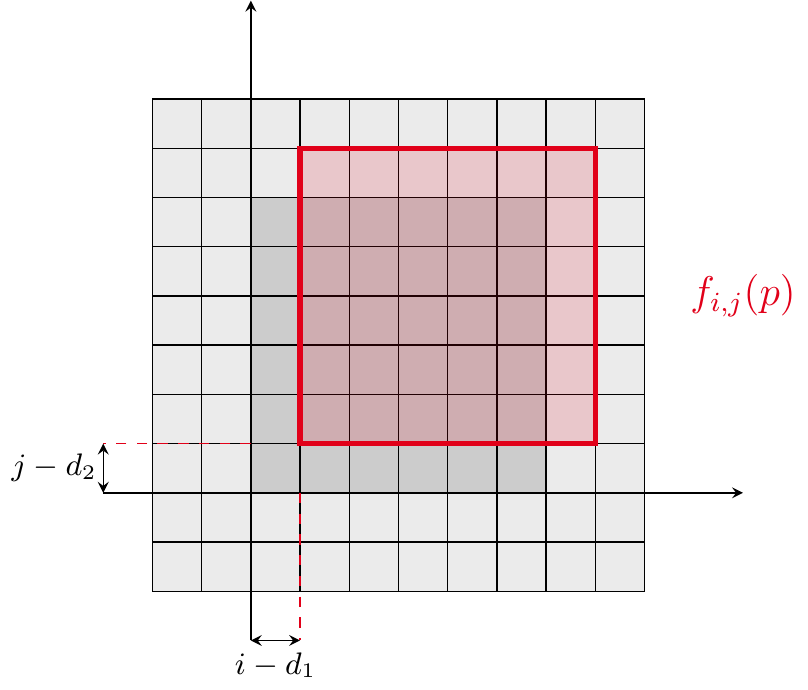}
         \caption{$\sigma^k(p^{ext})$ and $f_{i,j}(p) = \tau_{i-d_1, j-d_2}(\sigma^k(p^{ext}))|_{ \{0,\dots, n-1\}^2}$}
         \label{fig:spext}
     \end{subfigure}
        \caption{Illustration of the definition of $f_{i,j}(p)$ with $n=6, i=j=1, M=2, k=1$ and $(d_1, d_2) = (0,0)$. Figure (a) shows ${p \in \mathcal{P}_c(\lceil \frac{n-i}{M^k} \rceil, \lceil \frac{n-j}{M^k} \rceil)}$, (b) shows the completion to $p^{ext}$ and (c) the image of $p^{ext}$ by $\sigma^k$ and the pattern $f_{i,j}(p)$ extracted from it.  }
        \label{fig:fij}
\end{figure}

\noindent Then, the following function:
\[\begin{array}{ccccc}
f & : & \bigsqcup_{0 \leq i,j \leq M^k-1} \mathcal{P}_c(\lceil \frac{n-i}{M^k} \rceil, \lceil \frac{n-j}{M^k} \rceil) & \to & \mathcal{P}_c(n,n) \\
 & & p \in \mathcal{P}_c(\lceil \frac{n-i}{M^k} \rceil, \lceil \frac{n-j}{M^k} \rceil) & \mapsto & f_{i,j}(p) \\
\end{array}\]
 is well defined. Now we show that it is injective.

Let $p \in  \mathcal{P}_c(\lceil \frac{n-i}{M^k} \rceil, \lceil \frac{n-j}{M^k} \rceil)$, $p' \in  \mathcal{P}_c(\lceil \frac{n-i'}{M^k} \rceil, \lceil \frac{n-j'}{M^k} \rceil)$ and suppose that $f(p) = f(p')$. 
By definition of $f$, there are some patterns $p^{ext}, p'^{ext} \sqsubset c$, such that \[ p^{ext}|_{ \{0, \dots, \lceil \frac{n-i}{M^k} \rceil - 1\} \times \{0, \dots, \lceil \frac{n-j}{M^k} \rceil - 1\}} = p 
\text{~ and ~} 
p'^{ext}|_{ \{0, \dots, \lceil \frac{n-i'}{M^k} \rceil - 1\} \times \{0, \dots, \lceil \frac{n-j'}{M^k} \rceil - 1\}} = p' \] 
and that:
\begin{equation}
    \label{eq:fp}
    f(p) = \tau_{i-d_1,j-d_2}(\sigma^k(p^{ext}))|_{\{0, \dots, n-1\}^2} = \tau_{i' - d_1,j'-d_2}(\sigma^k(p'^{ext}))|_{\{0, \dots, n-1\}^2} = f(p')
\end{equation}

Moreover as $f(p) = f(p')$ is a pattern of size $(n,n)$, it has a unique desubstitution by $\sigma^k$ by assumption. Therefore, $(i-d_1,j-d_2) \equiv (i'-d_1, j'-d_2) \mod (M^k,M^k)$. As $i,j,i',j' \in \{0,\dots, M^k -1\}$, it implies that $(i,j) = (i',j')$. 
\noindent Consequently, \cref{eq:fp} can be rewritten as
\[ \tau_{i-d_1,j-d_2}(\sigma^k(p^{ext}))|_{\{0, \dots, n-1\}^2} = \tau_{i - d_1,j-d_2}(\sigma^k(p'^{ext}))|_{\{0, \dots, n-1\}^2} .\]
That is,
\[ \sigma^k(p^{ext})|_{\{0, \dots, n-1\}^2 - (i-d_1,j-d_2)} = \sigma^k(p'^{ext})|_{\{0, \dots, n-1\}^2 - (i-d_1,j-d_2)} . \]
And as $p^{ext} \sqsubseteq p$,
\begin{equation}
    \label{eq:sk}
    \sigma^k(p)|_{\{0, \dots, n-1\}^2 - (i-d_1,j-d_2)} = \sigma^k(p')|_{\{0, \dots, n-1\}^2 - (i-d_1,j-d_2)} ,
\end{equation}
on the domain where $\sigma^k(p)$ is defined.

\begin{sloppypar}
Now we show that $p$ and $p'$ are the same on all their domain. 
Let~${\vec u \in \{0, \dots, \lceil \frac{n-i}{M^k} \rceil - 1\} \times \{0, \dots, \lceil \frac{n-j}{M^k} \rceil - 1\}}$. First, we have:
\end{sloppypar}
\[ M^k \vec u + (d_1, d_2) \in \{0, \dots, n-1\}^2 - (i-d_1,j-d_2) .\]
Therefore by using the fact that $(d_1,d_2)$ is a determining position, \cref{eq:sk} gives:
\begin{align*}
 & \sigma^k(p)_{ M^k \vec u+ (d_1, d_2)} = \sigma^k(p')_{ M^k \vec u+ (d_1, d_2)}, \\
\Rightarrow~ & \sigma^k(p_{\vec u})_{d_1,d_2} = \sigma^k(p'_{\vec u})_{d_1,d_2}, \\
\Rightarrow~ & p_{\vec u} = p'_{\vec u} .
\end{align*}

\noindent This being true for all $\vec u \in \{0, \dots, \lceil \frac{n-i}{M^k} \rceil - 1\} \times \{0, \dots, \lceil \frac{n-j}{M^k} \rceil - 1\}$, it eventually holds that $p = p'$. This concludes the proof of the injectivity of $f$.

The injectivity of $f$ yields the following inequality:

\[
\left| \mathcal{L}_{n,n}(c) \right| \geq  \left| \bigsqcup_{0 \leq i,j \leq M^k-1} \mathcal{P}_c(\lceil \frac{n-i}{M^k} \rceil, \lceil \frac{n-j}{M^k} \rceil) \right|  =  \sum_{0\leq i,j \leq M^k-1} \left| \mathcal{P}_c(\lceil \frac{n-i}{M^k} \rceil, \lceil \frac{n-j}{M^k} \rceil) \right|.
\]
Since $c$ is aperiodic and by \cref{coro:karimoutot}, for all $0 \leq i,j \leq M^k - 1$, it holds that:
 \begin{align*}
 \left| \mathcal{P}_c(\lceil \frac{n-i}{M^k} \rceil, \lceil \frac{n-j}{M^k} \rceil) \right| & \geq  \lceil \frac{n-i}{M^k} \rceil \times \lceil \frac{n-j}{M^k} \rceil + 1 \\
  & \geq \lceil \frac{n-i}{M^k} \rceil \times \lceil \frac{n-j}{M^k} \rceil \left( 1 + \frac{1}{\lceil \frac{n-i}{M^k} \rceil \times \lceil \frac{n-j}{M^k} \rceil}\right)\\
  & \geq \lceil \frac{n-i}{M^k} \rceil \times \lceil \frac{n-j}{M^k} \rceil \left( 1 + \frac{1}{\lceil \frac{n}{M^k} \rceil ^2}\right)\\
 \end{align*}
Eventually, taking the sum and using a classical property of the ceiling function \cite[p85]{concretemath} gives:
\begin{align*}
 \left| \mathcal{L}_{n,n}(c) \right| & \geq \sum_{0\leq i,j \leq M^k-1} \left| \mathcal{P}_c(\lceil \frac{n-i}{M^k} \rceil, \lceil \frac{n-j}{M^k} \rceil) \right|\\
 & \geq \left( 1 + \frac{1}{\lceil \frac{n}{M^k} \rceil ^2}\right) \sum_{0\leq i,j \leq M^k-1} \lceil \frac{n-i}{M^k} \rceil \times \lceil \frac{n-j}{M^k} \rceil \\
 & =  \left( 1 + \frac{1}{\lceil \frac{n}{M^k} \rceil ^2}\right) n^2 .
\end{align*}
\end{proof}

Our last lemma gives us a bound on how many times patterns from a substitutive subshift can be desubstitutued thanks to Solomyak result's (\cref{theo:solomyak}).

\begin{lemma}
 \label{lemma:unsubstitution}
 Let $\sigma$ be an invertible substitution.
 If all patterns $p\in \mathcal{L}(X^\sigma)$ whose support contains $\{0, \dots, \rho - 1\}^2$ can be uniquely desubstitued by $\sigma$;
 then for all $k$, all patterns in $\mathcal{L}(X^\sigma)$ whose support contains $\{0, \dots, (\rho + 1)M^{k-1} -2\}^2$ can be uniquely desubstited by $\sigma^k$.
\end{lemma}

\begin{proof}
We show the lemma by induction on $k$. 
The result is true for $k = 1$. 
Suppose that the result is true for some $k \geq 1$ fixed. 
Let $p$ be a pattern in $\mathcal{L}(X^\sigma)$ with support $\mathcal{D}$ such that $\{0,\dots, (\rho + 1)M^k - 2 \}^2 \subset \mathcal{D}$. 
Let $\vec t, \vec t'\in\Z^2$ be such that $p  = \tau_{\vec t}(\sigma^{k+1}(c))|_{ \mathcal{D}} = \tau_{\vec t'}(\sigma^{k+1}(c'))|_{ \mathcal{D}}$ for some $c,c' \in X^\sigma$. 
By supposition, and because $(\rho + 1)M^k \geq \rho$, the pattern $p$ can be uniquely desubstitued by $\sigma$. This, and the fact that $p  = \tau_{\vec t}(\sigma(\sigma^{k}(c)))|_{ \mathcal{D}} = \tau_{\vec t'}(\sigma(\sigma^{k}(c')))|_{ \mathcal{D}}$ implies that $\vec t \equiv \vec t' \mod (M,M)$. 
Therefore $\vec t' = \vec t + M\vec u$ for some $\vec u \in \Z^2$ and then:
\[ \tau_{\vec t}(\sigma(\sigma^{k}(c)))|_{ \mathcal{D}} = \tau_{\vec t + M\vec u}(\sigma(\sigma^{k}(c')))|_{ \mathcal{D}} ,\]
\[ \Rightarrow \sigma(\sigma^{k}(c))|_{ \mathcal{D} + \vec t} = \tau_{au}(\sigma(\sigma^{k}(c')))|_{ \mathcal{D} + \vec t} \]
\[ \Rightarrow \sigma(\sigma^{k}(c))|_{ \mathcal{D} + \vec t} = \sigma(\tau_u(\sigma^{k}(c')))|_{ \mathcal{D} + \vec t} .\]
Let $\mathcal{R} = \{0, \dots, (\rho + 1)M^{k-1} -2\}^2$.
For some $\vec v \in \Z^2$, we have the following inclusion (see \cref{fig:subsets_lemma11}):
\begin{equation}
  \label{eq:inclusion_lemma11}
  M (\mathcal{R} + \vec v) + \{0, \dots, M-1\}^2\subset \{0, \dots, (\rho + 1)M^{k} -2\}^2 + \vec t \subset \mathcal{D} + \vec t .
\end{equation}

\begin{figure}[ht]
    \centering
    \includegraphics[scale=0.9]{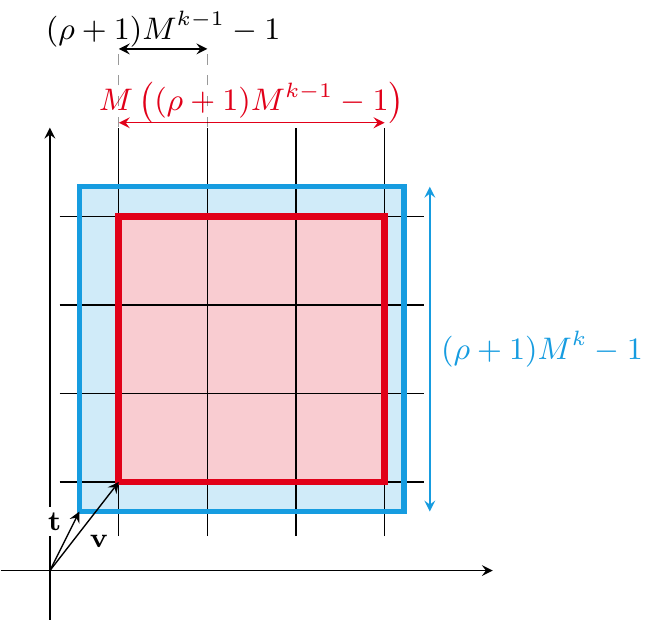}
    \caption{Inclusion of the sets of \cref{eq:inclusion_lemma11}.}
    \label{fig:subsets_lemma11}
\end{figure}

\noindent
Then, because $\sigma$ is invertible, we have:
\begin{align*}
  & \sigma(\sigma^{k}(c))|_{ \mathcal{D} + \vec t} = \sigma(\tau_u(\sigma^{k}(c')))|_{ \mathcal{D} + \vec t} \\
 \Rightarrow~ & \sigma(\sigma^{k}(c))|_{ M (\mathcal{R} + \vec v)+ \{0, \dots, M-1\}^2} = \sigma(\tau_u(\sigma^{k}(c')))|_{M (\mathcal{R} + \vec v)+ \{0, \dots, M-1\}^2} \\
 \Rightarrow~ & \sigma^{k}(c)|_{\mathcal{R} + \vec v} = \tau_u(\sigma^{k}(c'))|_{\mathcal{R} + \vec v} \\
\end{align*}
By induction hypothesis, this pattern can be uniquely desubstitued by $\sigma^k$, therefore it holds that $\vec u  \equiv 0 \mod M^k$, which implies that $\vec t' = \vec t + M\vec u \equiv \vec t \mod M^{k+1}$ and that $p$ can be uniquely desubstitued by $\sigma^{k+1}$. 
\end{proof}

One can then desubstitute patterns from a substitutive subshifts as many times as possible, and bringing the two previous lemmas together gives the lower bound:

\begin{theorem}
    \label{th:main}
	Let $\sigma$ be a primitive square aperiodic substitution having a determining position. 
	Then there exists $K > 1$ such that
	\[ \forall c \in X^\sigma, \forall n \in \N, \left| \mathcal{L}_{n,n}(c) \right| \geq K n^2 \]
\end{theorem}	

\begin{proof}
 By \cref{theo:solomyak}, there exists $\rho$ such that all patterns $p$ appearing in the configurations of $X^\sigma$ in which a square of size $\rho$ fits can be uniquely desubstitued by $\sigma$. 
 Let $c \in X^\sigma$ and $n \in \N$. Let $k(n)  = \max \{ k  \in \N \mid \forall p \in \mathcal{L}_{n,n}(c), p \text{ can be uniquely desubstitued by }\sigma^k \}$. 
 By \cref{turbolemme} it holds that:
 \[ |\mathcal{L}_{n,n}(c)| \geq n^2 \left( 1 + \frac{1}{\lceil \frac{n}{M^{k(n)}} \rceil^2} \right) \]

 Additionnaly, $(\rho + 1)M^{k(n)} - 1 > n$ otherwise \cref{lemma:unsubstitution} would contradict the maximality of $k(n)$.
 Therefore, $\rho + 1 > \frac{n}{M^{k(n)}} $ and $\rho + 1 > \lceil \frac{n}{M^{k(n)}} \rceil$. This turns our bound into:
 
  \[ |\mathcal{L}_{n,n}(c)| \geq n^2 \left( 1 + \frac{1}{(\rho + 1)^2} \right) \]
  
  \noindent which gives a constant $K := 1 + \frac{1}{(\rho + 1)^2} > 1$ which does depend neither on $c$ nor on $n$.
\end{proof}

\begin{remark}
    A similar bound can be found for non-square uniform substitutions.
    However, if one wants to know the pattern complexity $|\mathcal{L}_{n,n}(c)|$ for some $m,n$, the bound depends on how close the shape of the substitution is with respect to the rectangle $(m,n)$.
    More precisely, let $\sigma$ be a substitution of size $(M,N)$. Then for any $R>0$ such that
    \[ \left| \log_M(m) - \log_N(n) \right| \leq R ,\]
    there exists $K_R>1$ such that for all $c\in X^\sigma$,
    \[ |\mathcal{L}_{n,n}(c)| \geq K_R mn  . \]
\end{remark}

To conclude this section, let us apply our result to Robinson's tileset.
Let $X^R\subseteq A^{\Z^2}$ be the SFT constituted of all valid tiling by Robinson's tileset $A$ \cite{Robinson}. 

A natural way of constructing a substitution on $A$ is to ``grow'' the pattern of each tile to a $2\times 2$ pattern. 
Each tile is composed of a back ground marking (in light-blue and gray on our drawings) and of a main marking (in black). We generate the rules of our $2\times 2$ substitution $\sigma$ in the following way:
\begin{enumerate}
\item The image of each tile has the same background marking,
\item the upper right tile of the image of a tile $t$ has the same main marking than $t$,
\item the lower left tile of the image of $t$ has a main marking with four out going arrows having its blue parts oriented as the background marking of $t$,
\item the two other tiles are chosen to have one main arrows and to incoming arrows on each side so that the main arrow extends the upper right tile.
\end{enumerate}
\cref{fig:robinson_subs} illustrates the construction of $\sigma(t)$ for one tile $t\in A$ and \cref{fig:squares} its iteration.
For space considerations, we do not write the full substitution obtained by this process.
\begin{figure}[ht]
    \centering
    \includegraphics[width = 14cm]{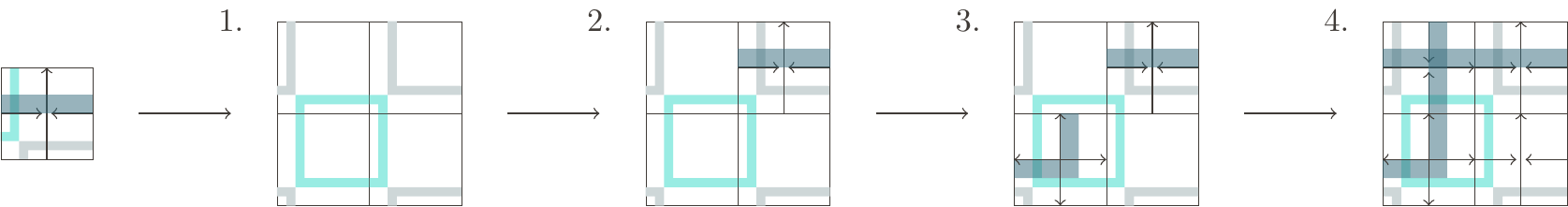}
    \caption{The four steps of construction of the substitution $\sigma$.} 
    \label{fig:robinson_subs}
\end{figure}

\begin{figure}[ht]
    \centering
    \includegraphics[width = 13cm]{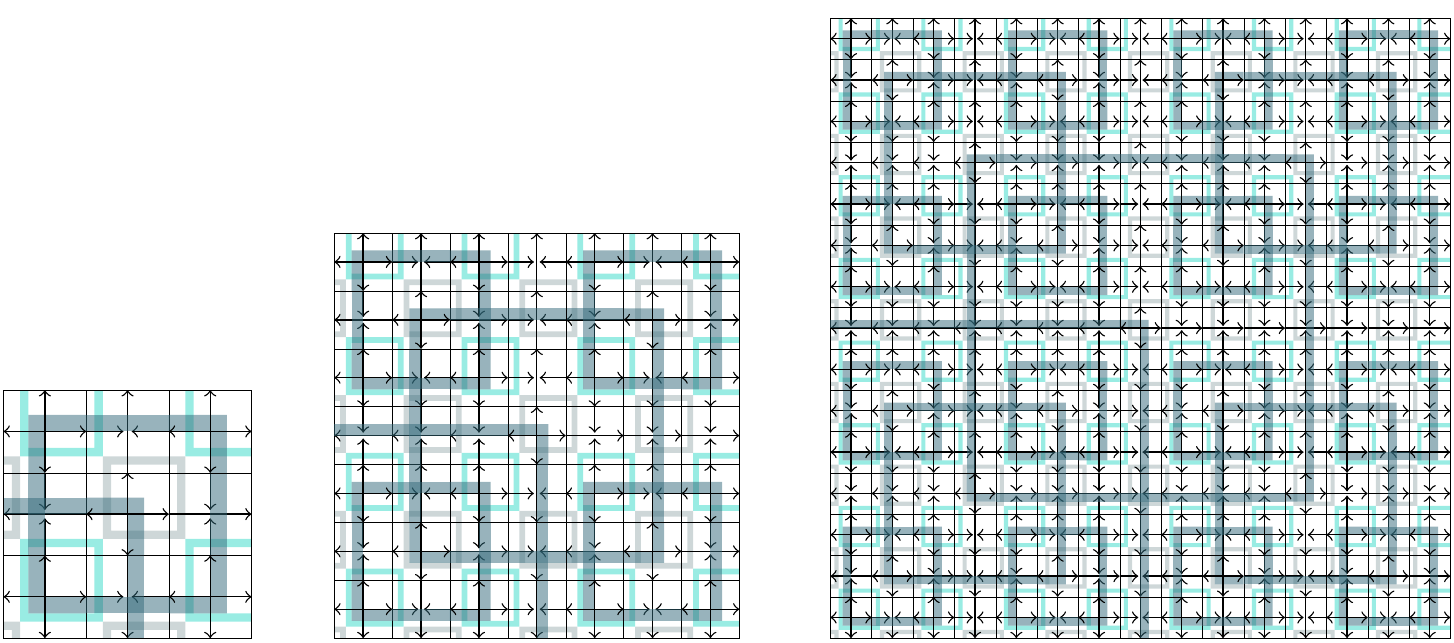}
    \caption{Portions of patterns generated by successive applications of $\sigma$.}
    \label{fig:squares}
\end{figure}

Even though $X^R$ is not a substitutive subshift by itself, as substitutive patterns of its configurations can be separated by “fracture lines”, one can show that $X^\sigma \subset X^R$, and therefore for all $n$, $|\mathcal{L}_{n,n}(X^\sigma)| \leq |\mathcal{L}_{n,n}(X^R)|$.

Now, we partition $A$ into two sets $A_1$ and $A_2$. $A_1$ contains all the tiles with four outgoing arrows, of the form \includegraphics[]{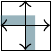}, and $A_2$ all the other tiles.
Then, we define 
\[
f: \begin{cases}
A &\rightarrow \{ 0,1 \} \\
t &\mapsto
\begin{cases}
 1 & \text{ if } t\in A_1\\
 0 & \text{ if } t\in A_2
\end{cases}
\end{cases} .
\]

$f$ is compatible with $\sigma$ in the sense that if $f(t_1) = f(t_2)$ for two tiles in $A$, then $f(\sigma(t_1)) = f(\sigma(t_2))$. 
Therefore the following substitution is well defined:
\[ \sigma':
\begin{cases}
\{0,1\} & \rightarrow \{0,1\}^{\{0,1\}^2} \\
a & \mapsto f((\sigma(t)) \text{ with } \sigma(t) = a
\end{cases}
\]
\cref{fig:robinson_proj} illustrate $\sigma'$ the kind of patters it generates.

\begin{figure}[ht]
    \centering
    \includegraphics[width = 9cm]{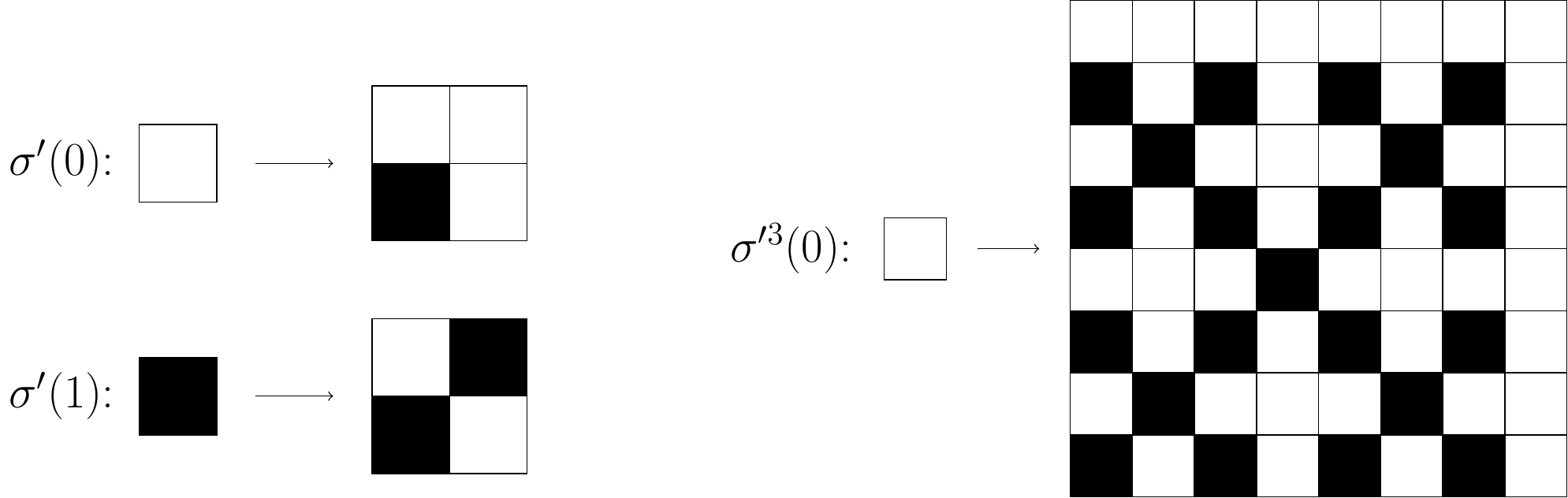}
    \caption{The substitution $\sigma'$ and an example of pattern generated by it.}
    \label{fig:robinson_proj}
\end{figure}

Let $X^{\sigma'}$ be the subshift of $\sigma'$. We have that $\mathcal{L}(X^{\sigma'}) = \{\sigma'^k(0), \sigma'^k(1) \mid k \in \N \} = \{ f(\sigma^k(t)) \mid t\in A, k \in \N \}$. Therefore for all $n$, it holds that $\mathcal{L}_{n,n}(X^{\sigma'}) = f(\mathcal{L}_{n,n}(X^{\sigma}))$. 
This implies that  $|\mathcal{L}_{n,n}(X^{\sigma'})| \leq |\mathcal{L}_{n,n}(X^{\sigma})|$. 
\noindent
Then, $\sigma'$ is primitive with a determining position and $X^{\sigma'}$ is aperiodic since every $1$ corresponds to one corner of a square in Robinson's tiling, therefore \cref{th:main} gives that there exists $K>1$ such that $|\mathcal{L}_{n,n}(X^{\sigma'})| \geq Kn^2 $, and then:
\[ |\mathcal{L}_{n,n}(X^R)| \geq K n^2 . \]

\section{Perspectives}
\label{sec:future}
We believe that \cref{th:main} can be extended to any square uniform primitive substitution, without the assumption of determining position. 
In the current proof, this condition allows us to ``recover'' information on patterns when de-substituting in \cref{turbolemme}, even when the pattern does not contain a full $\sigma^k(a)$ for some $a\in A$.
For now, we do not see how to obtain a quadratic lower bound without this information, but we think that it is possible to do.
Such a generalization would show that all substitutive subshifts of square primitive substitutions have complexity $|\mathcal{L}_{n,n}(X^\sigma)| \in \Theta(n^2)$.
In addition to a better understanding of aperiodic subshifts, we hope that this is the first step towards a classification of two-dimensional substitutions in terms of pattern complexity, similar to Pansiot's \cite{Pansiot} in the one dimensional case.

It is also unknown if the $mn+1$ complexity bound of \cref{coro:karimoutot} can be improved for non-substititive aperiodic subhifts.

\section*{Acknowledgements}
The authors want to thank Nathalie Aubrun for her guidance and lots of helpful advice, as well as Guilhem Gamard for his careful proofreading.

\newpage
\bibliography{biblio}

\newpage
\appendix

\section{The Robinson tileset}
\begin{figure}[ht]
\centering
\includegraphics{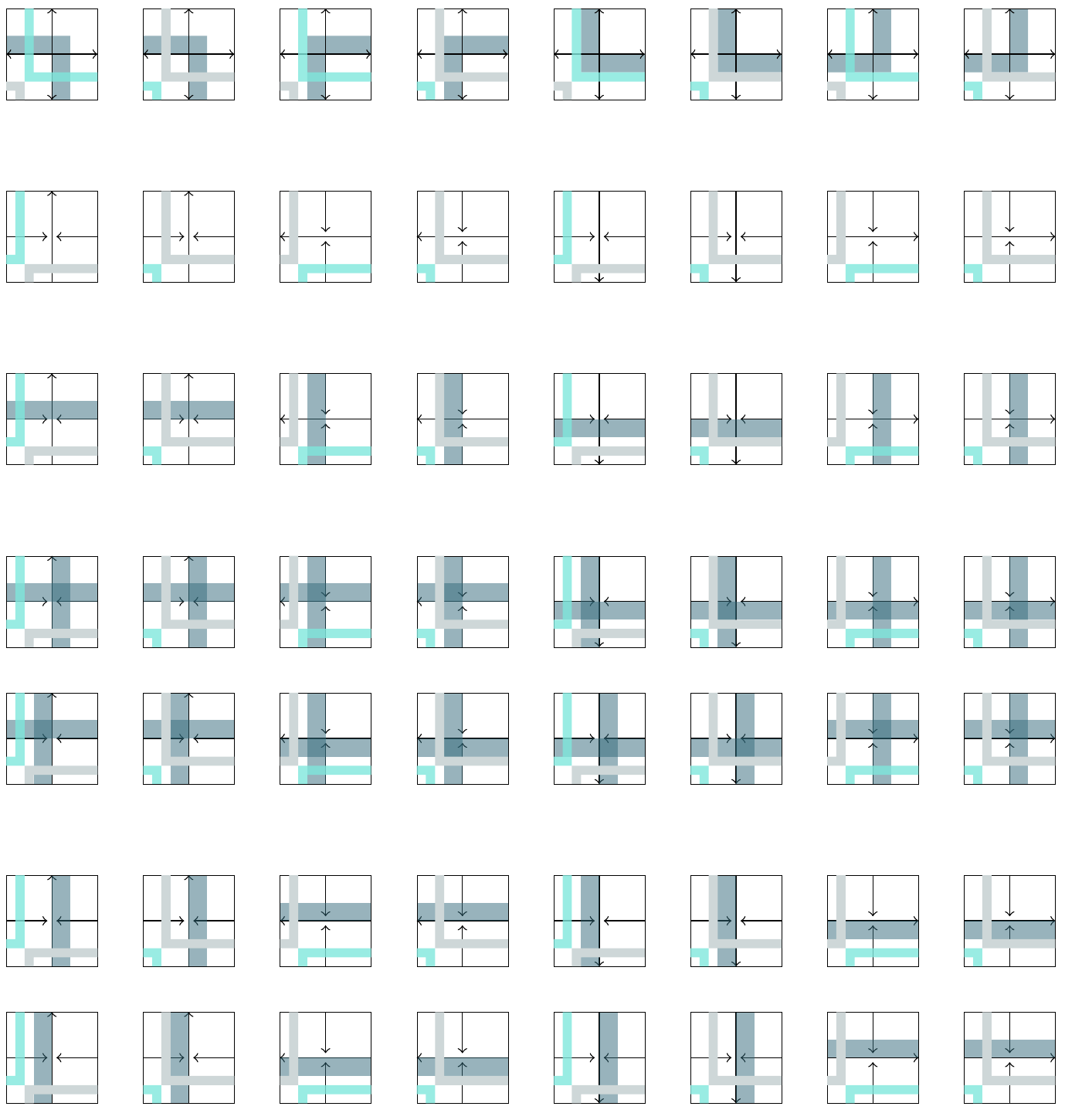}
\caption{Robinson's tileset.}
\label{fig:robinson_full}
\end{figure}

\end{document}